\documentclass[11pt]{article}
\usepackage[a4paper]{geometry}
\usepackage{graphicx, amsfonts, amsmath, amssymb, amsthm, graphicx, caption, authblk, url, multirow, makecell, framed, float, tikz, xcolor}
\usepackage[T1]{fontenc}
\theoremstyle{definition}
\newtheorem{theorem}{Theorem}

\newtheorem{example}{Example}

\newtheorem{lemma}{Lemma}

\theoremstyle{remark}

\begin{document}
\tikzstyle{line} = [draw]
\tikzstyle{arrow} = [draw, ->]
\tikzstyle{vertex} = [draw, circle, minimum size=0.8cm]
\tikzstyle{vertexg} = [draw, circle, minimum size=0.8cm, fill=black!25]
\title{Unpopularity Factor in the Marriage and Roommates Problems\thanks{A preliminary version of this paper \cite{ruangwises2} has appeared at CSR 2019.}}
\author[1]{Suthee Ruangwises\thanks{\texttt{ruangwises.s.aa@m.titech.ac.jp}}}
\author[1]{Toshiya Itoh\thanks{\texttt{titoh@c.titech.ac.jp}}}
\affil[1]{Department of Mathematical and Computing Science, Tokyo Institute of Technology, Tokyo, Japan}
\date{}
\maketitle

\begin{abstract}
Given a set $A$ of $n$ people, with each person having a \textit{preference list} that ranks a subset of $A$ as his/her acceptable partners in order of preference, we consider the \textit{Roommates Problem} (\textsc{rp}) and the \textit{Marriage Problem} (\textsc{mp}) of matching people with their partners. In \textsc{rp} there is no further restriction, while in \textsc{mp} only people of opposite genders can be acceptable partners. For a pair of matchings $X$ and $Y$, let $\phi(X,Y)$ denote the number of people who prefer a person they get matched by $X$ to a person they get matched by $Y$, and define an \textit{unpopularity factor} $u(M)$ of a matching $M$ to be the maximum ratio $\phi(M',M) / \phi(M,M')$ among all other possible matchings $M'$. In this paper, we develop an algorithm to compute the unpopularity factor of a given matching in $O(m\sqrt{n}\log^2 n)$ time for \textsc{rp} and in $O(m\sqrt{n}\log n)$ time for \textsc{mp}, where $m$ is the total length of people's preference lists. We also generalize the notion of unpopularity factor to a weighted setting where people are given different voting weights and show that our algorithm can be slightly modified to support that setting with the same running time.

\end{abstract}

\section{Introduction}
The \textit{Stable Marriage Problem} is one of the most actively studied problems in theoretical computer science and economics \cite{gusfield,roth}. It has many real-world applications including assignments of medical residents \cite{roth0} and high-school students \cite{abdulkadiroglu1,abdulkadiroglu2}. In the original bipartite setting called \textit{Marriage Problem} (\textsc{mp}), a set of $n/2$ men and a set of $n/2$ women are given. Each person has a \textit{preference list} that ranks all people of opposite gender in strict order of preference. A man $m$ and a woman $w$ are called a \textit{blocking pair} for a matching $M$ if they are not matched with each other in $M$ but prefer each other to their own partners in $M$. A matching is called \textit{stable} if it does not admit any blocking pair. Gale and Shapley \cite{gale} proved that a stable matching always exists in any instance and developed an $O(n^2)$ algorithm to find one. Their algorithm can be adapted to a setting where each person's preference list may not contain all people of opposite gender. It runs in $O(m)$ time in this setting, where $m$ is the total length of people's preference lists \cite{gusfield}.

The \textit{Stable Roommates Problem} is a generalization of the original Stable Marriage Problem to a non-bipartite setting called \textit{Roommates Problem} (\textsc{rp}), where each person can be matched with anyone regardless of gender. Unlike in \textsc{mp}, a stable matching in \textsc{rp} does not always exist. Irving \cite{irving} developed an $O(n^2)$ algorithm to find a stable matching in a given \textsc{rp} instance, or report that none exists.

\subsection{Popular Matchings}
Apart from stability, another well-studied property of a preferable matching is popularity. For a pair of matchings $X$ and $Y$, let $\phi(X,Y)$ denote the number of people who prefer a person they get matched by $X$ to a person they get matched by $Y$. A matching $M$ is called \textit{popular} if $\phi(M,M') \geq \phi(M',M)$ for any other matching $M'$. The concept of popularity of a matching was first introduced by Gardenfors \cite{gardenfors} in the context of the original Stable Marriage Problem. He also proved that in an \textsc{mp} instance where each person's preference list is \textit{strict} (containing no tie), every stable matching must be popular (but not vice-versa), hence a popular matching always exists.

The problem of determining whether a popular matching exists in a given instance, however, becomes more computationally challenged in other settings. Bir\'{o} et al. \cite{biro} proved that when ties among people in the preference lists are allowed, the problem of determining whether a popular matching exists in a given \textsc{mp} or \textsc{rp} instance is NP-hard. Very recently, Faenza et al. \cite{faenza} and Gupta et al. \cite{gupta} independently proved that this problem is still NP-hard for \textsc{rp} even when people's preference lists are strict. Cseh and Kavitha \cite{cseh} showed that in a complete graph \textsc{rp} instance where each person's preference list is strict and contains all other people, the problem of determining whether a popular matching exists is solvable in polynomial time for an odd $n$ but is NP-hard for an even $n$.

Problems related to popular matchings were also extensively studied in the setting of one-sided preference lists (matching each person with a unique item, where each person has a list that ranks items but each item does not have a list that ranks people) called \textit{House Allocation Problem} (\textsc{hap}). Abraham et al. \cite{abraham} developed an algorithm to find a popular matching in a given \textsc{hap} instance, or report that none exists. The algorithm runs in $O(m+n)$ time when people's preference lists are strict and in $O(m\sqrt{n})$ time when ties are allowed, where $m$ is the total length of people's preference lists and $n$ is the total number of people and items. Mestre \cite{mestre} later generalized their algorithm to a weighted setting where people are given different voting weights, while Manlove and Sng \cite{manlove} generalized it to a setting where each item is allowed to be matched with more than one person called \textit{Capacitated House Allocation Problem} (\textsc{chap}). Mahdian \cite{mahdian} studied the randomized version of this problem where people's preference lists are strict, \textit{complete} (containing all items), and randomly generated. He showed that a popular matching exists with high probability in a random \textsc{hap} instance if the ratio of the number of items to the number of people is greater than a specific constant. Ruangwises and Itoh \cite{ruangwises} later generalized Mahdian's study to the case where preference lists are strict but not complete, and found a similar behavior of the probability of existence of a popular matching. Abraham and Kavitha \cite{abraham-kavitha} proved that in any instance with at least one popular matching, one can achieve a popular matching by conducting at most two majority votes to force a change in assignments, starting at any matching. Kavitha et al. \cite{kavitha} introduced the concept of a \textit{mixed matching}, which is a probability distribution over a set of matchings, and proved that a mixed matching that is ``popular'' always exists.

\subsection{Unpopularity Measures}
While a popular matching may not exist in some instances, several measures of badness  of a matching that is not popular have been introduced. In the one-sided preference lists setting, McCutchen \cite{mccutchen} introduced two such measures: the \textit{unpopularity factor} and the \textit{unpopularity margin}. The unpopularity factor $u(M)$ of a matching $M$ is the maximum ratio $\phi(M',M) / \phi(M,M')$ among all other possible matchings $M'$, while the unpopularity margin $g(M)$ is the maximum difference $\phi(M',M) - \phi(M,M')$ among all other possible matchings $M'$. Note that the two measures are not equivalent as $\phi(M',M)$ and $\phi(M,M')$ may not add up to the total number of people since some people may like $M$ and $M'$ equally, thus it is possible for a matching to have higher unpopularity factor but lower unpopularity margin than another matching. See Example 1.

McCutchen \cite{mccutchen} developed an algorithm to compute $u(M)$ and $g(M)$ of a given matching $M$ of an $\textsc{hap}$ instance in $O(m\sqrt{n_2})$ and $O((g+1)m\sqrt{n})$ time, respectively, where $n_2$ is the number of items and $g = g(M)$ is the unpopularity margin of $M$. He also proved that the problem of finding a matching that minimizes either measure is NP-hard. Huang et al. \cite{huang1} later developed an algorithm to find a matching with bounded values of these measures in $\textsc{hap}$ instances with certain properties.

The notions of unpopularity factor and unpopularity margin also apply to the setting of two-sided preference lists (matching people with people). Bir\'{o} et al. \cite{biro} developed an algorithm to determine whether a given matching $M$ is popular in $O(m\sqrt{n \alpha(n,m)}\log^{3/2} n)$ time for \textsc{rp}, where $\alpha$ is the inverse Ackermann function (later improved to $O(m\sqrt{n}\log n)$ time when running with the recent fastest algorithm to find a maximum weight perfect matching \cite{duan}), and in $O(m\sqrt{n})$ time for \textsc{mp}. Their algorithm also simultaneously computes the unpopularity margin of $M$ during the run. Huang and Kavitha \cite{huang} proved that an \textsc{rp} instance with strict preference lists always has a matching with unpopularity factor $O(\log n)$, and it is NP-hard to find a matching with the lowest unpopularity factor, or even the one with less than $4/3$ times of the optimum.

\begin{example}
Consider the following \textsc{rp} instance. A set in a preference list means all people in that set are ranked equally, e.g. $a_2$ prefers $a_1$ and $a_4$ equally as his first choices over $a_3$.

\begin{figure}[H]
    \centering
    \begin{minipage}{0.3\textwidth}
        \underline{Preference Lists} \\
        $\boldsymbol{a_1:} \hspace{0.2cm} a_4, a_2, a_3$ \\
        $\boldsymbol{a_2:} \hspace{0.2cm} \{a_1, a_4\}, a_3$ \\
        $\boldsymbol{a_3:} \hspace{0.2cm} \{a_1, a_4\}, a_2$ \\
        $\boldsymbol{a_4:} \hspace{0.2cm} \{a_2, a_3\}, a_1$ \\
    \end{minipage}
    \begin{minipage}{0.3\textwidth}
        $M_0 = \{\{a_1,a_2\}, \{a_3,a_4\}\}$ \\
				$M_1 = \{\{a_1,a_3\}, \{a_2,a_4\}\}$ \\
        $M_2 = \{\{a_1,a_4\}, \{a_2,a_3\}\}$ \\
    \end{minipage}
\end{figure}

In this example, $\phi(M_0,M_1) = 1$, $\phi(M_1,M_0) = 0$, $\phi(M_0,M_2) = 3$, $\phi(M_2,M_0) = 1$, $\phi(M_1,M_2) = 3$, and $\phi(M_2,M_1) = 1$. Therefore, $M_0$ is popular, while $u(M_1) = \infty$, $g(M_1) = 1-0 = 1$, $u(M_2) = 3/1 = 3$, and $g(M_2) = 3-1 = 2$. Observe that $M_1$ has higher unpopularity factor but lower unpopularity margin than $M_2$. \qed
\end{example}

\subsection{Our Contribution}
The algorithm of Bir\'{o} et al. \cite{biro} determines whether a given matching $M$ is popular and also simultaneously computes the unpopularity margin of $M$, hence we currently have an algorithm to compute an unpopularity margin of a given matching in $O(m\sqrt{n}\log n)$ time for \textsc{rp} and in $O(m\sqrt{n})$ time for \textsc{mp}. However, there is currently no efficient algorithm to compute an unpopularity factor of a given matching in \textsc{mp} or \textsc{rp}.

In this paper, by employing an auxiliary graph similar to the one in \cite{biro}, we develop an algorithm to compute the unpopularity factor of a given matching. The algorithm runs in $O(m\sqrt{n}\log^2 n)$ time for \textsc{rp} and in $O(m\sqrt{n}\log n)$ time for \textsc{mp}. We also generalize the notion of unpopularity factor to the weighted setting where people are given different voting weights, and show that our algorithm can be slightly modified to support that setting with the same running time.

\section{Preliminaries}
Let $I$ be an \textsc{rp} or \textsc{mp} instance consisting of a set $A=\{a_1,...,a_n\}$ of $n$ people, with each person having a preference list that ranks a subset of $A$ as his/her acceptable partners in order of preference. In \textsc{rp} there is no further restriction, while in \textsc{mp} people are classified into two genders, and each person's preference list can contain only people of opposite gender. Throughout this paper, we consider a more general setting where ties among two or more people are allowed in the preference lists. Also, let $m$ be the total length of people's preference lists.

For a matching $M$ and a person $a \in A$, let $M(a)$ be the person matched with $a$ in $M$ (for convenience, let $M(a) = null$ if $a$ is unmatched in $M$). Also, let $r_a(b)$ be the rank of a person $b$ in $a$'s preference list, with the most preferred item(s) having rank 1, the second most preferred item(s) having rank 2, and so on (for convenience, let $r_a(null) = \infty$).

For any pair of matchings $X$ and $Y$, we define $\phi(X,Y)$ to be the number of people who strictly prefer the person they get matched by $X$ to the person they get matched by $Y$, i.e.
$$\phi(X,Y) = |\{a \in A|r_a(X(a)) < r_a(Y(a))\}|.$$
Also, let
$$
\Delta(X,Y) = \begin{cases}
\phi(Y,X) / \phi(X,Y), &\text{if } \phi(X,Y) > 0; \\
1, &\text{if } \phi(X,Y) = \phi(Y,X) = 0; \\
\infty, &\text{otherwise.}
\end{cases}
$$
Finally, define an unpopularity factor
$$u(M) = \max_{M' \in \mathbb{M}-\{M\}} \Delta(M,M'),$$
where $\mathbb{M}$ is the set of all matchings of a given instance $I$. Note that a matching $M$ is popular if and only if $u(M) \leq 1$.

\section{Unweighted Setting}
We first consider an unweighted setting where every person has equal voting weight.
 
\subsection{\textsc{rp} Instances}
Let $I$ be an \textsc{rp} instance, $M$ be a matching of $I$, and $k$ be an arbitrary positive rational number. Beginning with a similar approach to \cite{biro}, we construct an undirected graph $H_{(M,k)}$ with vertices $A \cup A'$, where $A' = \{a'_1,...,a'_n\}$ is a set of ``copies'' of people in $A$. An edge $\{a_i,a_j\}$ exists if and only if $a_i$ is in $a_j$'s preference list and $a_j$ is in $a_i$'s preference list; an edge $\{a'_i,a'_j\}$ exists if and only if $\{a_i,a_j\}$ exists; an edge $\{a_i,a'_j\}$ exists if and only if $i=j$.

The major distinction of our algorithm is that we assign weights to edges of $H_{(M,k)}$ differently from \cite{biro}. For each pair of $i$ and $j$ with an edge $\{a_i,a_j\}$, define $\delta_{i,j}$ as follows.
$$
\delta_{i,j} = \begin{cases}
1, &\text{if } a_i \text{ is unmatched in } M \text{ or } a_i \text{ prefers } a_j \text{ to } M(a_i); \\
-k, &\text{if } a_i \text{ prefers } M(a_i) \text{ to } a_j; \\
0, &\text{if } \{a_i,a_j\} \in M \text{ or } a_i \text{ likes } a_j \text{ and } M(a_i) \text{ equally}.
\end{cases}
$$
For each pair of $i$ and $j$, we set the weights of both $\{a_i,a_j\}$ and $\{a'_i,a'_j\}$ to be $\delta_{i,j}+\delta_{j,i}$. Finally, for each edge $\{a_i,a'_i\}$, we set its weight to be $-2k$ if $a_i$ is matched in $M$, and 0 otherwise. See Example 2.

\begin{example}
Consider the following matching $M$ in an \textsc{rp} instance.

\begin{figure}[H]
    \centering
    \begin{minipage}{0.3\textwidth}
				\underline{Preference Lists} \\
        $\boldsymbol{a_1:} \hspace{0.2cm} a_2, a_3, a_4$ \\
        $\boldsymbol{a_2:} \hspace{0.2cm} a_3, a_1$ \\
        $\boldsymbol{a_3:} \hspace{0.2cm} a_1, a_2, a_4$ \\
        $\boldsymbol{a_4:} \hspace{0.2cm} a_1, a_3$ \\
				
				$M = \{(a_1,a_2), (a_3,a_4)\}$ \\
    \end{minipage}
    \begin{minipage}{0.5\textwidth}
    	\centering
				\begin{tabular}{|c|c||c|c|c|c|} \hline
            \multicolumn{2}{|c||}{\multirow{2}{*}{$\delta_{i,j}$}} & \multicolumn{4}{c|}{$j$} \\ \cline{3-6}
            \multicolumn{2}{|c||}{} & \hspace{0.1cm}1\hspace{0.1cm} & \hspace{0.1cm}2\hspace{0.1cm} & \hspace{0.1cm}3\hspace{0.1cm} & \hspace{0.1cm}4\hspace{0.1cm} \\ \hline \hline
            \multirow{4}{*}{\hspace{0.1cm}$i$\hspace{0.1cm}} & \hspace{0.1cm}1\hspace{0.1cm} &  & 0 & -2 & -2 \\ \cline{2-6}
            & \hspace{0.1cm}2\hspace{0.1cm} & 0 &  & 1 &  \\ \cline{2-6}
            & \hspace{0.1cm}3\hspace{0.1cm} & 1 & 1 &  & 0 \\ \cline{2-6}
            & \hspace{0.1cm}4\hspace{0.1cm} & 1 &  & 0 &  \\ \hline
        \end{tabular}
		\end{minipage}
		
		\begin{minipage}{0.5\textwidth}
			\centering
        \begin{tikzpicture}[node distance=1.386cm, auto]
            \node [circle] (a0) {};
            \node [vertex, left of=a0, node distance=0.8cm] (a1) {$a_1$};
            \node [vertex, above of=a0] (a2) {$a_2$};
            \node [vertex, right of=a0, node distance=0.8cm] (a3) {$a_3$};
            \node [vertex, below of=a0] (a4) {$a_4$};
            \node [circle, right of=a0, node distance=3.2cm] (b0) {};
            \node [vertex, right of=b0, node distance=0.8cm] (b1) {$a'_1$};
            \node [vertex, above of=b0] (b2) {$a'_2$};
            \node [vertex, left of=b0, node distance=0.8cm] (b3) {$a'_3$};
            \node [vertex, below of=b0] (b4) {$a'_4$};
    
            \path [line] (a1) -- node{0} (a2);
            \path [line] (a2) -- node{2} (a3);
            \path [line] (a3) -- node{0} (a4);
            \path [line] (a4) -- node{-1} (a1);
            \path [line] (a1) -- node{-1} (a3);
            \path [line] (b2) -- node{0} (b1);
            \path [line] (b3) -- node{2} (b2);
            \path [line] (b4) -- node{0} (b3);
            \path [line] (b1) -- node{-1} (b4);
            \path [line] (b3) -- node{-1} (b1);
            \path [line, bend right=90] (a1) to [out=-30, in=-150] node{-4} (b1);
            \path [line] (a2) -- node{-4} (b2);
            \path [line] (a3) -- node{-4} (b3);
            \path [line] (a4) -- node{-4} (b4);
        \end{tikzpicture}
	\end{minipage}
\end{figure}

The values of all $\delta_{i,j}$ are shown in the top-right table, and the auxiliary graph $H_{(M,2)}$ is shown at the bottom. \qed
\end{example}

The intuition of constructing this auxiliary graph is that we want to check whether $u(M) > k$, i.e. whether there exists another matching $M'$ with the number of people who prefer $M'$ to $M$ more than $k$ times the number of those who prefer $M$ to $M'$. Each matching $M'$ is represented by a perfect matching of $H_{(M,k)}$ consisting of the edges of $M'$ in $A$ as well as their copies in $A'$, with each unmatched person $a_i$ being matched with his own copy $a'_i$. We want to add 1 for each person who prefers $M'$ to $M$ and subtract $k$ for each one who prefers $M$ to $M'$, and then check whether the sum is positive. This is the reason that we set $\delta_{i,j}$ to be $-k$ if $a_i$ prefers $M(a_i)$ to $a_j$, and the weight of $\{a_i,a'_i\}$ to be $-2k$ if $a_i$ is unmatched in $M$.

The relation between $u(M)$ and the graph $H_{(M,k)}$ is formally shown in the following lemma.

\begin{lemma} \label{lem1}
$u(M) > k$ if and only if $H_{(M,k)}$ contains a positive weight perfect matching.
\end{lemma}

\begin{proof}
For any matching $M'$, define $A_1(M')$ to be a set of people in $A$ that are matched in $M'$, and $A_2(M')$ to be a set of people in $A$ that are unmatched in $M'$. Also, define
\begin{align*}
A^+_1(M') &= \{a_i \in A_1(M')|a_i\text{ is unmatched in }M\text{ or }a_i\text{ prefers }M'(a_i)\text{ to }M(a_i)\}; \\
A^-_1(M') &= \{a_i \in A_1(M')|a_i\text{ prefers }M(a_i)\text{ to }M'(a_i)\}; \\
A^-_2(M') &= \{a_i \in A_2(M')|a_i\text{ is matched in }M\}.
\end{align*}
We have $\phi(M',M) = |A^+_1(M')|$ and $\phi(M,M') = |A^-_1(M')|+|A^-_2(M')|$.
						
Suppose that $u(M) > k$. From the definition of $u(M)$, there must be a matching $M_0$ such that $\phi(M_0,M) > k\phi(M,M_0)$. In the graph $H_{(M,k)}$, consider a perfect matching
$$S_0 = M_0 \cup \{\{a'_i,a'_j\}|\{a_i,a_j\} \in M_0\} \cup \{\{a_i,a'_i\}|a_i\text{ is unmatched in }M_0\}$$
with weight $W_0$. From the definition, we have
\begin{align*}
W_0 &= 2\left(|A^+_1(M_0)|-k|A^-_1(M_0)|\right)-2k|A^-_2(M_0)| \\
&= 2\left(|A^+_1(M_0)|-k\left(|A^-_1(M_0)|+|A^-_2(M_0)|\right)\right) \\
&= 2(\phi(M_0,M) - k\phi(M,M_0)) \\
&> 0,
\end{align*}
hence $H_{(M,k)}$ contains a positive weight perfect matching.
						
On the other hand, suppose there is a positive weight perfect matching $S_1$ of $H_{(M,k)}$ with weight $W_1$. See Example 3. Let $M_1 =$ $\{\{a_i,a_j\} \in S_1\}$ and $M_2 = \{\{a_i,a_j\}|\{a'_i,a'_j\} \in S_1\}$. Since $S_1$ is a perfect matching of $H_{(M,k)}$, we have $A_2(M_1) = A_2(M_2)$, and
\begin{align*}
0 &< W_1 \\
&= \left(|A^+_1(M_1)|-k|A^-_1(M_1)|\right) + \left(|A^+_1(M_2)|-k|A^-_1(M_2)|\right) - 2k|A^-_2(M_1)| \\
&= \left(|A^+_1(M_1)|-k|A^-_1(M_1)|\right) + \left(|A^+_1(M_2)|-k|A^-_1(M_2)|\right) - k|A^-_2(M_1)|-k|A^-_2(M_2)| \\
&= (\phi(M_1,M) - k\phi(M,M_1)) + (\phi(M_2,M) - k\phi(M,M_2)).
\end{align*}
Therefore, we have either $\phi(M_1,M) > k\phi(M,M_1)$ or $\phi(M_2,M) > k\phi(M,M_2)$, which implies $u(M) > k$.
\end{proof}

\begin{example}
Consider the auxiliary graphs $H_{(M,2)}$ and $H_{(M,3)}$ constructed from a matching $M$ in Example 2.

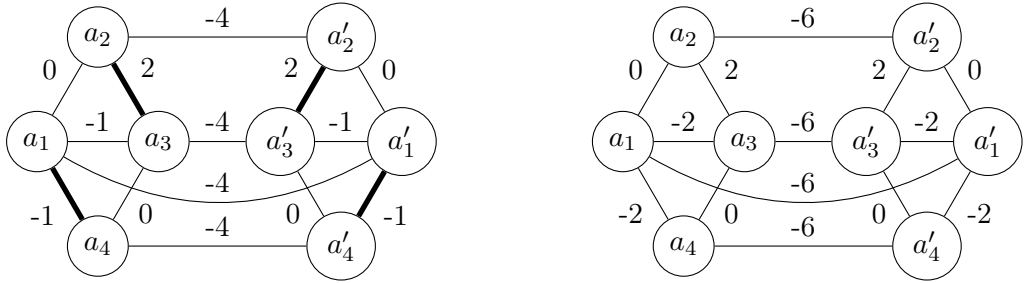
\begin{figure}[H]
	\begin{minipage}{0.5\textwidth}
		\centering
        \begin{tikzpicture}[node distance=1.386cm, auto]
            \node [circle] (a0) {};
            \node [vertex, left of=a0, node distance=0.8cm] (a1) {$a_1$};
            \node [vertex, above of=a0] (a2) {$a_2$};
            \node [vertex, right of=a0, node distance=0.8cm] (a3) {$a_3$};
            \node [vertex, below of=a0] (a4) {$a_4$};
            \node [circle, right of=a0, node distance=3.2cm] (b0) {};
            \node [vertex, right of=b0, node distance=0.8cm] (b1) {$a'_1$};
            \node [vertex, above of=b0] (b2) {$a'_2$};
            \node [vertex, left of=b0, node distance=0.8cm] (b3) {$a'_3$};
            \node [vertex, below of=b0] (b4) {$a'_4$};
            
            \path [line] (a1) -- node{0} (a2);
            \path [line, line width=2pt] (a2) -- node{2} (a3);
            \path [line] (a3) -- node{0} (a4);
            \path [line, line width=2pt] (a4) -- node{-1} (a1);
            \path [line] (a1) -- node{-1} (a3);
            \path [line] (b2) -- node{0} (b1);
            \path [line, line width=2pt] (b3) -- node{2} (b2);
            \path [line] (b4) -- node{0} (b3);
            \path [line, line width=2pt] (b1) -- node{-1} (b4);
            \path [line] (b3) -- node{-1} (b1);
            \path [line, bend right=90] (a1) to [out=-30, in=-150] node{-4} (b1);
            \path [line] (a2) -- node{-4} (b2);
            \path [line] (a3) -- node{-4} (b3);
            \path [line] (a4) -- node{-4} (b4);
        \end{tikzpicture}
	\end{minipage}
	\begin{minipage}{0.5\textwidth}
		\centering
        \begin{tikzpicture}[node distance=1.386cm, auto]
            \node [circle] (a0) {};
            \node [vertex, left of=a0, node distance=0.8cm] (a1) {$a_1$};
            \node [vertex, above of=a0] (a2) {$a_2$};
            \node [vertex, right of=a0, node distance=0.8cm] (a3) {$a_3$};
            \node [vertex, below of=a0] (a4) {$a_4$};
            \node [circle, right of=a0, node distance=3.2cm] (b0) {};
            \node [vertex, right of=b0, node distance=0.8cm] (b1) {$a'_1$};
            \node [vertex, above of=b0] (b2) {$a'_2$};
            \node [vertex, left of=b0, node distance=0.8cm] (b3) {$a'_3$};
            \node [vertex, below of=b0] (b4) {$a'_4$};
            
            \path [line] (a1) -- node{0} (a2);
            \path [line] (a2) -- node{2} (a3);
            \path [line] (a3) -- node{0} (a4);
            \path [line] (a4) -- node{-2} (a1);
            \path [line] (a1) -- node{-2} (a3);
            \path [line] (b2) -- node{0} (b1);
            \path [line] (b3) -- node{2} (b2);
            \path [line] (b4) -- node{0} (b3);
            \path [line] (b1) -- node{-2} (b4);
            \path [line] (b3) -- node{-2} (b1);
            \path [line, bend right=90] (a1) to [out=-30, in=-150] node{-6} (b1);
            \path [line] (a2) -- node{-6} (b2);
            \path [line] (a3) -- node{-6} (b3);
            \path [line] (a4) -- node{-6} (b4);
        \end{tikzpicture}
	\end{minipage}
\end{figure}

On the left, $H_{(M,2)}$ has a positive weight perfect matching consisting of the bold-faced edges, but on the right, $H_{(M,3)}$ does not. This implies $2 < u(M) \leq 3$. \qed
\end{example}

For a given value of $k$, the problem of determining whether $u(M) > k$ is now transformed to detecting a positive weight perfect matching of $H_{(M,k)}$, which can be done by finding the maximum weight perfect matching of $H_{(M,k)}$.

\begin{lemma} \label{lem2}
Given an \textsc{rp} instance $I$, a matching $M$ of $I$, and a rational number $k = x/y$, where $x \in [0,n-1]$ and $y \in [1,n]$ are integers, there is an algorithm to determine whether $u(M) > k$ in $O(m\sqrt{n}\log n)$ time.
\end{lemma}

\begin{proof}
From Lemma \ref{lem1}, the problem of determining whether $u(M) > k$ is equivalent to determining whether $H_{(M,k)}$ has a positive weight perfect matching. Observe that $H_{(M,k)}$ has $O(n)$ vertices and $O(m)$ edges, and we can multiply the weights of all edges by $y$ so that they are all integers with magnitude $O(n)$. Using the recent algorithm of Duan et al. \cite{duan}, we can find a maximum weight perfect matching in a graph with integer weight edges of magnitude $\text{poly}(n)$ in $O(m\sqrt{n}\log n)$ time, hence we can detect a positive weight perfect matching in $H_{(M,k)}$ in $O(m\sqrt{n}\log n)$ time.
\end{proof}

As the possible values of $u(M)$ are limited, we can perform a binary search for its value. This allows us to efficiently compute $u(M)$. To the best of our knowledge, this is the first approach on popular matchings that employs the binary search technique.

\begin{theorem} \label{thm1}
Given an \textsc{rp} instance $I$ and a matching $M$ of $I$, there is an algorithm to compute $u(M)$ in $O(m\sqrt{n}\log^2 n)$ time.
\end{theorem}

\begin{proof}
Observe that if $u(M)$ is not $\infty$, it must be in the form of $x/y$, where $x \in [0,n-1]$ and $y \in [1,n]$ are integers, meaning that there are at most $n^2$ possible values of $u(M)$. By performing a binary search on the value of $k = x/y$ (if $u(M) > n-1$, then $u(M) = \infty$), we run the algorithm in Lemma \ref{lem2} to determine whether $u(M) > k$ for $O(\log n^2) = O(\log n)$ times to find the exact value of $u(M)$, hence the total running time is $O(m\sqrt{n}\log^2 n)$.
\end{proof}

\subsection{\textsc{mp} Instances}
The running time of the algorithm in Theorem \ref{thm1} is for a general \textsc{rp} instance. However, in an \textsc{mp} instance we can improve it using the following approach. For any matching $M$ in an \textsc{mp} instance, we define a matching
$$S = M \cup \{\{a'_i,a'_j\}|\{a_i,a_j\} \in M\} \cup \{\{a_i,a'_i\}|a_i\text{ is unmatched in }M\}$$
in the graph $H_{(M,k)}$. Since $S$ is a perfect matching, for any perfect matching $S'$ of $H_{(M,k)}$, every edge of $S'$ that is not in $S$ must be a part of some cycle in which the edges alternate between $S$ and $S'$. Moreover, from the definition of $\delta_{i,j}$, every edge of $S$ has zero weight. Therefore, $H_{(M,k)}$ contains a positive weight perfect matching if and only if it contains a positive weight alternating cycle relative to $S$. Hence, the problem becomes equivalent to detecting a positive weight alternating cycle (relative to $S$) in $H_{(M,k)}$. Note that this property holds for every \textsc{rp} instance, not limited to only \textsc{mp}.

However, the special property of \textsc{mp} is that $A$ is bipartite. Let $A_M$ and $A_W$ be the two parts of $A$ with no edge between vertices in the same part (which correspond to the sets of men and women, respectively). Also, let $A'_M = \{a'_i|a_i \in A_M\}$ and $A'_W = \{a'_i|a_i \in A_W\}$. Observe that we can divide the vertices of $H_{(M,k)}$ into two parts $H_1 = A_M \cup A'_W$ and $H_2 = A_W \cup A'_M$ with no edge between vertices in the same part, so $H_{(M,k)}$ is also bipartite. In $H_{(M,k)}$, we orient the edges of $S$ toward $H_2$ and all other edges toward $H_1$, hence the problem of detecting a positive weight alternating cycle becomes equivalent to detecting a positive weight directed cycle (see Example 4), which can be done in $O(m\sqrt{n})$ time using the shortest path algorithm of Goldberg \cite{goldberg}. Therefore, by performing a binary search on the value of $u(M)$ similar to in \textsc{rp}, the total running time for \textsc{mp} is $O(m\sqrt{n}\log n)$.

\begin{example}
Consider the following matching $M'$ in an \textsc{mp} instance with men $a_1$ and $a_3$, and women $a_2$ and $a_4$.

\begin{figure}[H]
	\hspace{0.18\textwidth}
	\begin{minipage}{0.3\textwidth}
	\underline{Preference Lists} \\
        $\boldsymbol{a_1:} \hspace{0.2cm} a_2, a_4$ \\
				$\boldsymbol{a_2:} \hspace{0.2cm} a_3, a_1$ \\
        $\boldsymbol{a_3:} \hspace{0.2cm} a_2, a_4$ \\
        $\boldsymbol{a_4:} \hspace{0.2cm} a_1, a_3$ \\ \\
	\end{minipage}
	\begin{minipage}{0.34\textwidth}
			$A_M = \{a_1,a_3\}$ \\
			$A_W = \{a_2,a_4\}$ \\ \\
			$M' = \{\{a_1,a_2\}, \{a_3,a_4\}\}$ \\ \\
	\end{minipage}
	
	\begin{minipage}{0.5\textwidth}
		\centering
        \begin{tikzpicture}[node distance=1.386cm, auto]
            \node [circle] (a0) {};
            \node [vertex, left of=a0, node distance=0.8cm] (a1) {$a_1$};
            \node [vertex, above of=a0] (a2) {$a_2$};
            \node [vertex, right of=a0, node distance=0.8cm] (a3) {$a_3$};
            \node [vertex, below of=a0] (a4) {$a_4$};
            \node [circle, right of=a0, node distance=3.2cm] (b0) {};
            \node [vertex, right of=b0, node distance=0.8cm] (b1) {$a'_1$};
            \node [vertex, above of=b0] (b2) {$a'_2$};
            \node [vertex, left of=b0, node distance=0.8cm] (b3) {$a'_3$};
            \node [vertex, below of=b0] (b4) {$a'_4$};
            
            \path [line, dashed] (a1) -- node{0} (a2);
            \path [line, line width=2pt] (a2) -- node{2} (a3);
            \path [line, dashed] (a3) -- node{0} (a4);
            \path [line, line width=2pt] (a4) -- node{-1} (a1);
            \path [line, dashed] (b2) -- node{0} (b1);
            \path [line, line width=2pt] (b3) -- node{2} (b2);
            \path [line, dashed] (b4) -- node{0} (b3);
            \path [line, line width=2pt] (b1) -- node{-1} (b4);
            \path [line, bend right=90] (a1) to [out=-30, in=-150] node{-4} (b1);
            \path [line] (a2) -- node{-4} (b2);
            \path [line] (a3) -- node{-4} (b3);
            \path [line] (a4) -- node{-4} (b4);
        \end{tikzpicture}
	\end{minipage}
	\begin{minipage}{0.5\textwidth}
		\centering
        \begin{tikzpicture}[node distance=1.386cm, auto]
            \node [circle] (a0) {};
            \node [vertex, left of=a0, node distance=0.8cm] (a1) {$a_1$};
            \node [vertexg, above of=a0] (a2) {$a_2$};
            \node [vertex, right of=a0, node distance=0.8cm] (a3) {$a_3$};
            \node [vertexg, below of=a0] (a4) {$a_4$};
            \node [circle, right of=a0, node distance=3.2cm] (b0) {};
            \node [vertexg, right of=b0, node distance=0.8cm] (b1) {$a'_1$};
            \node [vertex, above of=b0] (b2) {$a'_2$};
            \node [vertexg, left of=b0, node distance=0.8cm] (b3) {$a'_3$};
            \node [vertex, below of=b0] (b4) {$a'_4$};
            
            \path [arrow, dashed, line width=2pt] (a1) -- node{0} (a2);
            \path [arrow, line width=2pt] (a2) -- node{2} (a3);
            \path [arrow, dashed, line width=2pt] (a3) -- node{0} (a4);
            \path [arrow, line width=2pt] (a4) -- node{-1} (a1);
            \path [arrow, dashed] (b2) -- node{0} (b1);
            \path [arrow] (b3) -- node{2} (b2);
            \path [arrow, dashed] (b4) -- node{0} (b3);
            \path [arrow] (b1) -- node{-1} (b4);
            \path [arrow, bend left=90] (b1) to [out=30, in=150] node[above]{-4} (a1);
            \path [arrow] (a2) -- node{-4} (b2);
            \path [arrow] (b3) -- node[above]{-4} (a3);
            \path [arrow] (a4) -- node{-4} (b4);
        \end{tikzpicture}
	\end{minipage}
\end{figure}

On the left, $H_{(M',2)}$ has a positive weight perfect matching consisting of the bold-faced edges, while $S$ consists of the dotted edges.

On the right, since $H_{(M',2)}$ is a bipartite graph with parts $H_1 = \{a_1,a_3,a'_2,a'_4\}$ (white vertices) and $H_2 = \{a_2,a_4,a'_1,a'_3\}$ (gray vertices), we orient the edges of $S$ (dotted arrows) toward $H_2$, and the rest toward $H_1$. This directed graph has a positive weight directed cycle consisting of the bold-faced arrows. Both figures imply $u(M') > 2$. \qed
\end{example}

In a way similar to \textsc{rp}, we have the following lemma and theorem for \textsc{mp}.

\begin{lemma} \label{lem3}
Given an \textsc{mp} instance $I$, a matching $M$ of $I$, and a number $k = x/y$, where $x \in [0,n-1]$ and $y \in [1,n]$ are integers, there is an algorithm to determine whether $u(M) > k$ in $O(m\sqrt{n})$ time.
\end{lemma}

\begin{theorem} \label{thm2}
Given an \textsc{mp} instance $I$ and a matching $M$ of $I$, there is an algorithm to compute $u(M)$ in $O(m\sqrt{n}\log n)$ time.
\end{theorem}

\section{Weighted Setting}
The previous section shows the algorithm to compute an unpopularity factor of a given matching in an unweighted \textsc{rp} or \textsc{mp} instance where every person has equal voting weight. However, in many real-world situations, people may have different voting weights based on position, seniority, etc. Our algorithm can also be slightly modified to support a weighted instance with integer weights bounded by $N = \text{poly}(n)$ with the same running time in both \textsc{rp} and \textsc{mp}.

In the weighted setting, each person $a_i \in A$ has a weight $w(a_i)$. We analogously define $\phi(M,M')$ to be the sum of weights of people who strictly prefer a matching $M$ to a matching $M'$, i.e.
$$\phi(M,M') = \sum_{a \in A_{(M,M')}} w(a),$$
where ${A_{(M,M')} = \{a \in A|r_a(M(a)) < r_a(M'(a))\}}$. We also define $\Delta(M,M')$ and $u(M)$ the same way as in the unweighted setting. For each $a_i \in A$, we assume that $w(a_i)$ is a non-negative integer not exceeding $N = \text{poly}(n)$. Note that an unweighted instance can be viewed as a special case of a weighted instance where $w(a_i) = 1$ for all $a_i \in A$.

To support the weighted setting, we construct an auxiliary graph $H_{(M,k)}$ with the same set of vertices and edges as in the unweighted setting, but with slightly different weights of the edges. For each pair of $i$ and $j$ with an edge $\{a_i,a_j\}$, define
$$
\delta_{i,j} = \begin{cases}
w(a_i), &\text{if } a_i \text{ is unmatched in } M \text{ or } a_i \text{ prefers } a_j \text{ to } M(a_i); \\
-kw(a_i), &\text{if } a_i \text{ prefers } M(a_i) \text{ to } a_j; \\
0, &\text{if } \{a_i,a_j\} \in M \text{ or } a_i \text{ likes } a_j \text{ and } M(a_i) \text{ equally}.
\end{cases}
$$
For each pair of $i$ and $j$, the weights of $\{a_i,a_j\}$ and $\{a'_i,a'_j\}$ is $\delta_{i,j}+\delta_{j,i}$. Finally, for each edge $\{a_i,a'_i\}$, we set its weight to be $-2kw(a_i)$ if $a_i$ is matched in $M$, and 0 otherwise.

The auxiliary graph $H_{(M,k)}$ still has the same relation with $u(M)$, as shown in the following lemma.

\begin{lemma} \label{lemw1}
In the weighted \textsc{rp} instance, $u(M) > k$ if and only if $H_{(M,k)}$ contains a positive weight perfect matching.
\end{lemma}

\begin{proof}
The proof of this lemma is very similar to that of Lemma \ref{lem1}. We define the sets $A_1(M')$, $A_2(M')$, $A^+_1(M')$, $A^-_1(M')$, and $A^-_2(M')$ by the same way as in the proof of Lemma \ref{lem1}. However, from now on we will compute the sum of weights of each set's elements instead of counting the number of its elements.

For any set $B$, define $w(B) = \sum_{a \in B} w(a)$. We have $\phi(M',M) = w(A^+_1(M'))$ and $\phi(M,M') = w(A^-_1(M'))+w(A^-_2(M'))$.
						
Suppose that $u(M) > k$. There must exist a matching $M_0$ such that $\phi(M_0,M) > k\phi(M,M_0)$. Similarly to the proof of Lemma \ref{lem1}, in the graph $H_{(M,k)}$ consider a perfect matching
$$S_0 = M_0 \cup \{\{a'_i,a'_j\}|\{a_i,a_j\} \in M_0\} \cup \{\{a_i,a'_i\}|a_i\text{ is unmatched in }M_0\}$$
with weight $W_0$. From the definition, we have
\begin{align*}
W_0 &= 2\left(w(A^+_1(M_0))-kw(A^-_1(M_0))\right)-2kw(A^-_2(M_0)) \\
&= 2\left(w(A^+_1(M_0))-k\left(w(A^-_1(M_0))+w(A^-_2(M_0))\right)\right) \\
&= 2(\phi(M_0,M) - k\phi(M,M_0)) \\
&> 0,
\end{align*}
hence $H_{(M,k)}$ contains a positive weight perfect matching.
						
On the other hand, suppose there is a positive weight perfect matching $S_1$ of $H_{(M,k)}$ with weight $W_1$. Let $M_1 =$ $\{\{a_i,a_j\} \in S_1\}$ and $M_2 = \{\{a_i,a_j\}|\{a'_i,a'_j\} \in S_1\}$. Similarly to the proof of Lemma \ref{lem1}, we have $A_2(M_1) = A_2(M_2)$, and
\begin{align*}
0 &< W_1\\
&= \left(w(A^+_1(M_1))-kw(A^-_1(M_1))\right) + \left(w(A^+_1(M_2))-kw(A^-_1(M_2))\right) - 2kw(A^-_2(M_1)) \\
&= \left(w(A^+_1(M_1))-kw(A^-_1(M_1))\right) + \left(w(A^+_1(M_2))-kw(A^-_1(M_2))\right) \\
&~~~~~~~~~ - kw(A^-_2(M_1))-kw(A^-_2(M_2)) \\
&= (\phi(M_1,M) - k\phi(M,M_1)) + (\phi(M_2,M) - k\phi(M,M_2)).
\end{align*}
Therefore, we have either $\phi(M_1,M) > k\phi(M,M_1)$ or $\phi(M_2,M) > k\phi(M,M_2)$, which implies $u(M) > k$.
\end{proof}

Since the weights of people are bounded by $N = \text{poly}(n)$, the unpopularity factor $u(M)$ must be in the form $k = x/y$, where $x$ and $y$ are integers not exceeding $Nn$. For a given value of $k$, if we multiply the weights of all edges of $H_{(M,k)}$ by $y$, they will be integers with magnitude $O(Nn) = \text{poly}(n)$. Therefore, we can still use the algorithm of Duan et al. \cite{duan} to find a maximum weight perfect matching of $H_{(M,k)}$ with the same running time.

Moreover, there are at most $O(N^2n^2)$ possible values of $u(M)$. By performing a binary search on the value of $k$, we have to run the above algorithm for $O(\log N^2n^2) = O(\log n)$ times as in the unweighted setting, hence the total running time is still $O(m\sqrt{n}\log^2 n)$.

The argument for \textsc{mp} instances still works for the weighted setting as well since $H_{(M,k)}$ is still bipartite, hence we have the following theorems for the weighted setting \textsc{rp} and \textsc{mp}.

\begin{theorem} \label{thmw1}
Given a weighted \textsc{rp} instance $I$ with integer weights bounded by $N = \text{poly}(n)$ and a matching $M$ of $I$, there is an algorithm to compute $u(M)$ in $O(m\sqrt{n}\log^2 n)$ time.
\end{theorem}

\begin{theorem} \label{thmw2}
Given a weighted \textsc{mp} instance $I$ with integer weights bounded by $N = \text{poly}(n)$ and a matching $M$ of $I$, there is an algorithm to compute $u(M)$ in $O(m\sqrt{n}\log n)$ time.
\end{theorem}

\section{Concluding Remarks}
We develop an algorithm to compute the unpopularity factor of a given matching in $O(m\sqrt{n}$ $\log^2 n)$ time for \textsc{rp} and $O(m\sqrt{n}\log n)$ time for \textsc{mp}, which runs only slightly slower than the algorithm of McCutchen \cite{mccutchen} to solve the same problem in \textsc{hap} and the algorithm of Bir\'{o} et al. \cite{biro} to compute the unpopularity margin of a given matching in \textsc{rp} and \textsc{mp}. Our results also complete Tables \ref{table1} and \ref{table2}, which show the running time of the currently best known algorithms related to popularity in \textsc{rp}, \textsc{mp}, and \textsc{hap} in the unweighted setting with strict preference lists, and with ties allowed, respectively. In both tables, $m$ is the total length of preference lists, $n$ is the total number of people and items, $n_2$ is the number of items (for \textsc{hap}), and $g$ is the unpopularity margin of a given matching.

While the problem of finding a matching that minimizes the unpopularity factor or margin in a given matching is NP-hard, the problem of approximating the optimum of either measure is still open. For the unpopularity factor in \textsc{rp} with strict preference lists, the current best algorithm is the one developed by Huang and Kavitha \cite{huang}, which approximates it up to $O(\log n)$ factor. A possible future work is to investigate whether there is a better approximation algorithm for \textsc{rp}, or to develop one for \textsc{hap}. For the unpopularity margin, however, there is currently no efficient algorithm to approximate the optimum, both in \textsc{rp} and \textsc{hap}. Another possible future work is to investigate the probability of existence of a popular matching in \textsc{rp} where each person's preference list is independently and uniformly generated at random, similarly to the study of Mahdian \cite{mahdian}, and Ruangwises and Itoh \cite{ruangwises} in \textsc{hap}.

\begin{table}[H]
	\centering
	\begin{tabular}{|c|c|c|c|}
		\hline
		\multirow{2}{*}{} & \multicolumn{2}{c|}{\textbf{Two-sided Lists}} & \textbf{One-sided Lists} \\ \cline{2-4}
		& \textbf{\thead{Roommates\\ Problem (\textsc{rp})}} & \textbf{\thead{Marriage\\ Problem (\textsc{mp})}} & \textbf{\thead{House Allocation\\ Problem (\textsc{hap})}} \\ \hline
		\thead{Determine if a popular\\ matching exists} & \multirow{2}{*}{\thead{\\ NP-hard \cite{faenza,gupta}}} & \multirow{3}{*}{\thead{\\ \\ $O(m)$ \cite{gardenfors}}} & $O(m+n)$ \cite{abraham} \\ \cline{1-1} \cline{4-4}
		\thead{Find a matching $M$\\ that minimizes $g(M)$} & & & \multirow{2}{*}{\thead{\\ NP-hard \cite{mccutchen}}} \\ \cline{1-2}
		\thead{Find a matching $M$\\ that minimizes $u(M)$} & NP-hard \cite{huang} & & \\ \hline
		\thead{Test popularity\\ of a given matching $M$} & \multirow{2}{*}{\thead{\\ $O(m\sqrt{n}\log n)$ \cite{biro,duan}}} & \multirow{2}{*}{\thead{\\ $O(m\sqrt{n})$ \cite{biro}}} & $O(m+n)$ \cite{abraham} \\ \cline{1-1} \cline{4-4}
		\thead{Compute $g(M)$\\ of a given matching $M$} & & & $O((g+1)m\sqrt{n})$ \cite{mccutchen} \\ \hline
		\thead{Compute $u(M)$\\ of a given matching $M$} & \boldmath{$O(m\sqrt{n}\log^2 n)$} \textbf{[\S3]} & \boldmath{$O(m\sqrt{n}\log n)$} \textbf{[\S3]} & $O(m\sqrt{n_2})$ \cite{mccutchen} \\ \hline
	\end{tabular}
	\medskip
	\caption{Currently best known algorithms for an unweighted instance with strict preference lists}\label{table1}
\end{table}

\begin{table}[H]
	\centering
	\begin{tabular}{|c|c|c|c|}
		\hline
		\multirow{2}{*}{} & \multicolumn{2}{c|}{\textbf{Two-sided Lists}} & \textbf{One-sided Lists} \\ \cline{2-4}
		& \textbf{\thead{Roommates\\ Problem (\textsc{rp})}} & \textbf{\thead{Marriage\\ Problem (\textsc{mp})}} & \textbf{\thead{House Allocation\\ Problem (\textsc{hap})}} \\ \hline
		\thead{Determine if a popular\\ matching exists} & \multicolumn{2}{c|}{\multirow{3}{*}{\thead{\\ \\ NP-hard \cite{biro}}}} & $O(m\sqrt{n})$ \cite{abraham} \\ \cline{1-1} \cline{4-4}
		\thead{Find a matching $M$\\ that minimizes $g(M)$} & \multicolumn{2}{c|}{} & \multirow{2}{*}{\thead{\\ NP-hard \cite{mccutchen}}} \\ \cline{1-1}
		\thead{Find a matching $M$\\ that minimizes $u(M)$} & \multicolumn{2}{c|}{} & \\ \hline
		\thead{Test popularity\\ of a given matching $M$} & \multirow{2}{*}{\thead{\\ $O(m\sqrt{n}\log n)$ \cite{biro,duan}}} & \multirow{2}{*}{\thead{\\ $O(m\sqrt{n})$ \cite{biro}}} & $O(m\sqrt{n_2})$ \cite{mccutchen} \\ \cline{1-1} \cline{4-4}
		\thead{Compute $g(M)$\\ of a given matching $M$} & & & $O((g+1)m\sqrt{n})$ \cite{mccutchen} \\ \hline
		\thead{Compute $u(M)$\\ of a given matching $M$} & \boldmath{$O(m\sqrt{n}\log^2 n)$} \textbf{[\S3]} & \boldmath{$O(m\sqrt{n}\log n)$} \textbf{[\S3]} & $O(m\sqrt{n_2})$ \cite{mccutchen} \\ \hline
	\end{tabular}
	\medskip
	\caption{Currently best known algorithms for an unweighted instance with ties allowed in the preference lists}\label{table2}
\end{table}

\end{document}